\documentclass[a4paper,12pt]{article}

\usepackage{amssymb}
\usepackage{amsfonts}
\usepackage{graphicx}%[dvips]
\usepackage{amsmath}
\usepackage{latexsym}
\usepackage{pstricks}
\usepackage{mathrsfs}

\newtheorem{theorem}{Theorem}[section]
\newtheorem{lemma}[theorem]{Lemma}
\newtheorem{remark}[theorem]{Remark}

\newtheorem{corollary}[theorem]{Corollary}

\newenvironment{proof}[1][Proof]{\begin{trivlist}
\item[\hskip \labelsep {\bfseries #1}]}{\end{trivlist}}

\newcommand{\om}{\Omega}

\newcommand{\dc}{\|}

\newcommand{\iom}{\int_\om}

\newcommand{\nt}{(0,T)}

\newcommand{\omnt}{\om\times\nt}

\newcommand{\unu}{{\bfu}_0}

\newcommand{\bfx}{\mbox{\boldmath{$x$}}}
\newcommand{\bfn}{\mbox{\boldmath{$n$}}}
\newcommand{\bfe}{\mbox{\boldmath{$e$}}}
\newcommand{\bff}{\mbox{\boldmath{$f$}}}

\newcommand{\bfu}{\mbox{\boldmath{$u$}}}
\newcommand{\bfw}{\mbox{\boldmath{$w$}}}
\newcommand{\bfv}{\mbox{\boldmath{$v$}}}
\newcommand{\bfz}{\mbox{\boldmath{$z$}}}

\newcommand{\bfU}{\mbox{\boldmath{$U$}}}

\newcommand{\bfphi}{\mbox{\boldmath{$\phi$}}}

\newcommand{\bfzero}{{\bf 0}}
\newcounter{constants}
\begin{document}

\title{Strong Solutions to Non--Stationary Channel Flows of
Heat--Conducting Viscous Incompressible Fluids with Dissipative
Heating}

\author{
{Michal Bene\v{s}\footnote{This research was supported by the
project GA\v{C}R P201/10/P396. Additional support from the Ministry
of Education, Youth and Sports of the Czech Republic, project No.
1M0579, within activities of the CIDEAS research centre is greatly
acknowledged.}}
\bigskip
\\
{\small Department of Mathematics}
\\
{\small Centre for Integrated Design of Advanced Structures}
\\
{\small Faculty of Civil Engineering}
\\
{\small Czech Technical University in Prague}
\\
{\small Th\'{a}kurova 7, 166 29 Prague 6, Czech Republic}
\\
e-mail:  benes@mat.fsv.cvut.cz}

\date{}

\maketitle

\begin{abstract}
We study an initial-boundary-value problem for time-dependent flows
of heat-conducting viscous incompressible fluids in channel-like
domains on a time interval $(0,T)$. For the parabolic system with
strong nonlinearities and including the artificial (the so called
``do nothing'') boundary conditions, we prove the local in time
existence, global uniqueness and smoothness of the solution on a
time interval $(0,T^*)$, where $0< T^* \leq T$.
%\keywords{Navier-Stokes equations \and heat equation \and
%heat--conducting fluid \and qualitative properties \and mixed
%boundary conditions}
%\subclass{35Q30 \and 35K05 \and 76D03}
\end{abstract}

%%%%%%%%%%%%%%%%%%%%%%%%%%%%%%%%%%%%%%%%%%%%%%%%%%%%%%%%%%%%%%%%%%%%%%%

\section{Introduction}
\subsection{Preliminaries}

Let $\Omega \in \mathcal{C}^{0,1}$ be a two-dimensional bounded
domain with the boundary $\partial\Omega$.  Let $\partial\Omega =
\overline{\Gamma}_D \cup \overline{\Gamma}_N$ be such that
$\Gamma_D$ and $\Gamma_N$ are open, not necessarily connected, the
one-dimensional measure of $\Gamma_D \cap \Gamma_N$ is zero and
$\Gamma_D\neq\emptyset$ ($\Gamma_N=\bigcup_i^m \Gamma^{(i)}_N$,
$\overline{\Gamma}^{(i)}_N\cap\overline{\Gamma}^{(j)}_N=\emptyset$
for $i\neq j$). In a physical sense, $\Omega$ represents a
``truncated'' region of an unbounded channel system occupied by a
moving heat--conducting viscous incompressible fluid. $\Gamma_D$
will denote the ``lateral'' surface and $\Gamma_N$ represents the
open parts of the region $\Omega$. We assume that in/outflow channel
segments extend as straight pipes. All portions of $\Gamma_N$ are
taken to be flat and the boundary $\Gamma_N$ and rigid boundary
$\Gamma_D$ form a right angle at each point where the boundary
conditions change. Moreover, we assume that $\Gamma_D$ is smooth (of
class $\mathcal{C}^{\infty}$).

The flow of a viscous incompressible heat--conducting
constant--property fluid is governed by balance equations for linear
momentum, mass and internal energy \cite{diaz2007}
\begin{align}
\varrho\left(\bfu_t + (\bfu\cdot\nabla)\bfu \right) - \nu\Delta \bfu
+ \nabla \pi &= \varrho(1-\alpha\theta)\bff, && \label{momentum}
\\
{\rm div}\,\bfu &= 0, && \label{mass}
\\
c_p\varrho\left( \theta_t + \bfu\cdot\nabla\theta \right) - \kappa
\Delta \theta  -\nu \bfe(\bfu):\bfe(\bfu) &= \varrho\alpha
\theta\bff\cdot\bfu + h. && \label{energy}
\end{align}
Here $\bfu=(u_1,u_2)$, $\pi$ and $\theta$ denote the unknown
velocity, pressure and temperature, respectively. Tensor
$\bfe(\bfu)$ denotes the symmetric part of the velocity gradient.
Data of the problem are as follows: $\bff$ is a body force and $h$ a
heat source term. Positive constant material coefficients represent
the kinematic viscosity $\nu$, density $\varrho$, heat conductivity
$\kappa$, specific heat at constant pressure $c_p$ and thermal
expansion coefficient of the fluid $\alpha$. The energy balance
equation \eqref{energy} takes into account the phenomena of the
viscous energy dissipation and adiabatic heat effects. For rigorous
derivation of the model like \eqref{momentum}--\eqref{energy} we
refer the readers to \cite{Kagei2000}.

Concerning the boundary conditions of the flow, it is a standard
situation to prescribe the non-homogeneous Dirichlet boundary
condition for temperature and homogeneous no-slip boundary condition
for velocity of the fluid on the fixed walls of the channel. In the
case of temperature outflow boundary condition on $\Gamma_N$, we
accept a frequently used assumption as ``zero flux density'', which
is equivalent to the condition $\nabla\theta\cdot \bfn =0$,
sometimes referred to as a ``do nothing'' (or ``natural'') boundary
condition and rather used in numerical simulations (cf.
\cite{SaniGresho}). However, it is really not clear which boundary
condition for $\bfu$ should be prescribed on $\Gamma_N$. The
boundary condition
\begin{equation}\label{neumann}
-\pi \bfn_i+\nu\frac{\partial \bfu}{\partial \bfn_i}= F_i \bfn_i
\end{equation}
prescribed on $\Gamma^{(i)}_N$ is again often called ``do nothing''
(or ``free outflow'') boundary condition (cf. \cite{GaRaRoTu},
\cite{Gresho}, \cite{Hey}). Here $F_i$ are given functions of space
and time and $\bfn_i$ is the outer unit normal to $\Gamma^{(i)}_N$,
$i=1,\dots,m$. The boundary condition \eqref{neumann} results from a
variational principle and does not have a real physical meaning but
is rather used in truncating large physical domains. It has been
proven to be convenient in numerical modeling of parallel flows. For
more information about application of this boundary condition and
the physical meaning of the quantities $F_i$ we refer to
\cite{GaRaRoTu}.
\begin{remark}
Assume that $F_i$ are given smooth functions on $\Gamma^{(i)}_N$,
$i=1,\dots,m$, and consider the smooth extension $F$ such that
$F(\bfx,t)\big|_{\Gamma^{(i)}_N}\equiv F_i(\bfx,t)$. Introducing the
new variable $\mathcal{P}=\pi+F$ this amounts to solving the
homogeneous ``do nothing'' boundary condition transferring the data
from the right-hand side of \eqref{neumann} to the right-hand side
of the linear momentum balance equation. Hence we will assume
throughout this paper, without loss of generality, that $F_i\equiv
0$ in \eqref{neumann}, $i=1,\dots,m$.
\end{remark}
\begin{remark}
To simplify the notation in the whole paper, we normalize material
constants $\varrho$, $\nu$, $\kappa$, $\alpha$ and $c_p$ to one.
\end{remark}
The paper is organized as follows. In Subsection
\ref{notation_spaces}, we introduce basic notations and some
appropriate function spaces in order to precisely formulate our
problem. In Section \ref{strong_formulation}, we present the strong
form of the model for the non-stationary motion of viscous
incompressible heat-conducting fluids in a channel considered in our
work, impose compatibility conditions on initial data, specify our
smoothness assumptions on data and formulate the problem in a
variational setting. The main result of our work is established at
the end of Section \ref{strong_formulation}. In
Section~\ref{Auxiliary results}, we present basic results on the
existence, uniqueness and energy estimates of the solution to an
auxiliary linearized problems, the decoupled initial-boundary value
problems for the Stokes and heat equations. The main result, stated
in Section \ref{strong_formulation}, is proved in Section~\ref{sec
proof_main} via the Banach contraction principle. In the proof of
local in time existence, presented in Subsection
\ref{proof_existence}, we rely on the energy estimates for linear
problems, regularity of stationary solutions and interpolations-like
inequalities. The global in time uniqueness of the strong solution
is proved in Subsection \ref{sec_uniqueness} using the technique of
Gronwall lemma.
\begin{remark}
Let us note that our results can be extended to the problem if we
consider the so called ``free surface'' boundary condition on
$\Gamma_N$ and replace \eqref{neumann} by
$$-\pi\bfn+\nu[\nabla\bfu+(\nabla\bfu)^{\top}]\bfn={\bf0}.$$
However, to ensure the smoothness of the solution and exclude
boundary singularities near the points where the boundary conditions
change their type some additional requirements on the geometry of
the domain need to be introduced. This means that $\Gamma_N$ and
$\Gamma_D$ form an angle $\omega<\pi/4$ at each point where the
boundary conditions change (see \cite{OrSan}).
\end{remark}

%%%%%%%%%%%%%%%%%%%%%%%%%%%%%%%%%%%%%%%%%%%%%%%%%%%%%%%%%%%%%%%%%%%%%%

\subsection{Basic notation and some function
spaces}\label{notation_spaces}
Vector functions and operators acting on vector functions are
denoted by~boldface letters. Unless specified otherwise, we use
Einstein's summation convention for indices running from $1$ to $2$.
Throughout the paper, we will always use positive constants $c$,
$c_1$, $c_2$, $\dots$, which are not specified and which may differ
from line to line.

For an arbitrary $r\in [1,+\infty]$, $L^r(\Omega)$ denotes the usual
Lebesgue space equipped with the norm $\|\cdot\|_{L^r(\Omega)}$, and
$W^{k,p}(\Omega)$, $k\geq 0$ ($k$ need not to be an integer, see
\cite{KufFucJoh1977}), $p \in [1,+\infty]$, denotes the usual
Sobolev space with the norm $\|\cdot\|_{W^{k,p}(\Omega)}$. For
simplicity we denote shortly $\mathbf{W}^{k,p}\equiv
W^{k,p}(\Omega)^2$ and $\mathbf{L}^{r}\equiv L^{r}(\Omega)^2$.

To simplify mathematical formulations we introduce the following
notations:
\begin{align}
a_u(\bfu,\bfv) & :=   \iom \frac{\partial u_i}{\partial x_j}
\frac{\partial v_i}{\partial x_j}\,{\rm d}(\Omega),
\label{form_a} \\
b(\bfu,\bfv,\bfw)& := \iom u_j{\frac{\partial v_i}{\partial x_j}}
w_i\,{\rm d}(\Omega),\label{form_b}
 \\
a_{\theta}(\theta,\varphi) & :=  \iom \nabla \theta \cdot \nabla
\varphi \,{\rm d}(\Omega),\label{form_c}
 \\
d(\bfu,\theta,\varphi)& := \iom  \bfu\cdot\nabla\theta
 \,\varphi \,{\rm d}(\Omega),
 \label{form_d}
 \\
e(\bfu,\bfv,\varphi) & :=  \iom  e_{ij}(\bfu) e_{ij}(\bfv) \varphi
\,{\rm d}(\Omega), \label{form_e}
 \\
(\bfu,\bfv) &:=  \iom \bfu \cdot \bfv \,{\rm d}(\Omega),
 \label{scalar_Lu}
\\
(\theta,\varphi)_{\Omega} & :=  \iom \theta \varphi \,{\rm
d}(\Omega). \label{scalar_Lt}
\end{align}
In \eqref{form_a}--\eqref{scalar_Lt} all functions
$\bfu,\bfv,\bfw,\theta,\varphi$ are smooth enough, such that all
integrals on the right-hand sides make sense. In \eqref{form_e}
$e_{ij}(\bfu)$ denotes the components of the tensor $\bfe(\bfu)$
defined by
\begin{displaymath}
e_{ij}(\bfu)=\frac{1}{2}\left(\frac{\partial u_i}{\partial
x_j}+\frac{\partial u_j}{\partial x_i}\right), \qquad i,j=1,2.
\end{displaymath}
Let
\begin{equation}
\mathcal{E}_u := \left\{\bfu\in C^\infty(\overline{\Omega})^2; \,
\textmd{div}\,\bfu = 0, \, {\textmd{supp}\, \bfu}  \cap \Gamma_D =
\emptyset  \right\}
\end{equation}
and
\begin{equation}
\mathcal{E}_{\theta} := \left\{\theta \in
C^\infty(\overline{\Omega}); \, {\textmd{supp}\, \theta }  \cap
\Gamma_D = \emptyset\right\}
\end{equation}
and $\mathbf{V}_u^{k,p}$ be the closure of $\mathcal{E}_u$ in the
norm of $W^{k,p}(\om)^2$, $k\ge 0$ and $1\leq p \leq \infty$.
Similarly, let $V_{\theta}^{k,p}$ be a closure of
$\mathcal{E}_{\theta}$ in the norm of $W^{k,p}(\om)$. Then
$\mathbf{V}_{u}^{k,p}$ and $V_{\theta}^{k,p}$, respectively, are
Banach spaces with the norms of the spaces $W^{k,p}(\om)^2$ and
$W^{k,p}(\om)$, respectively. Note, that $\mathbf{V}_{u}^{1,2}$,
 $V_{\theta}^{1,2}$, $\mathbf{V}_{u}^{0,2}$ and $V_{\theta}^{0,2}$,
respectively, are Hilbert spaces with scalar products
\eqref{form_a}, \eqref{form_c}, \eqref{scalar_Lu} and
\eqref{scalar_Lt}, respectively.

Further, define the spaces
\begin{equation}
\mathcal{D}_u := \left\{\bfu \; |\; \bff \in \mathbf{V}_{u}^{0,2},\,
a_u(\bfu,\bfv) =(\bff,\bfv)  \textmd{ for all } \bfv\in
\mathbf{V}_{u}^{1,2} \right\} \label{D_u}
\end{equation}
and
\begin{equation}
\mathcal{D}_{\theta} :=  \left\{\theta \; |\; h \in
V_{\theta}^{0,2}, \, a_{\theta}(\theta,\varphi)
=(h,\varphi)_{\Omega} \textmd{ for all } \varphi \in
V_{\theta}^{1,2} \right\}\label{D_t},
\end{equation}
equipped with the norms
\begin{equation}\label{norm_D_u}
\|\bfu\|_{\mathcal{D}_u} := \|\bff\|_{\mathbf{V}_{u}^{0,2}}
\end{equation}
and
\begin{equation}\label{norm_D_t}
\|\theta\|_{\mathcal{D}_{\theta}} := \|h\|_{{V}_{\theta}^{0,2}},
\end{equation}
where $\bfu$ and $\bff$ are corresponding functions via \eqref{D_u}.
Similarly, $\theta$ and $h$ are corresponding functions via
\eqref{D_t}.

The key embeddings $\mathcal{D}_u\hookrightarrow \mathbf{W}^{2,2}$
and $\mathcal{D}_{\theta}\hookrightarrow {W}^{2,2}(\Omega)$ are
consequences of assumptions setting on the domain $\Omega$ and the
regularity results for the steady Stokes system in channel-like
domains with ``do-nothing'' condition, see \cite[Remark 2.2 and
Corollary 2.3]{Benes2009} and the ``classical'' regularity results
for the stationary linear heat equation (the Poisson equation) with
the mixed boundary conditions, see for instance \cite{KufSan}.

%%%%%%%%%%%%%%%%%%%%%%%%%%%%%%%%%%%%%%%%%%%%%%%%%%%%%%%%%%%%%%%%%%%%%

% ----------------------------------------------------------------
\section{Formulation of the problem and the main result}
\label{strong_formulation}
Let $T\in(0,\infty)$ be fixed throughout the paper and $Q_T =
\omnt$, $\Gamma_{DT} = \Gamma_D\times(0,T)$ and
$\Gamma_{NT}=\Gamma_N\times(0,T)$. The strong formulation of our
problem is as follows:
\begin{align}
\bfu_t + (\bfu\cdot\nabla)\bfu  - \Delta \bfu + \nabla
\mathcal{P}+\theta\bff &=\bff &&\textmd{in}\; Q_T, \label{eq3}
\\
{\rm div}\,\bfu &= 0 &&\textmd{in}\;Q_T, \label{eq5}
\\
 \theta_t + \bfu\cdot\nabla\theta  -
\Delta \theta  - \bfe(\bfu):\bfe(\bfu) &= \theta\bff\cdot\bfu + h &&
\textmd{in}\;Q_T, \label{heat equation}
\\
 \bfu &= {\bf 0}  &&\textmd{on}\;\Gamma_{DT},
\label{eq6}
\\
\theta&=g  &&\textmd{on}\;\Gamma_{DT},\label{boundary temperature}
\\
-\mathcal{P}\bfn+\frac{\partial\bfu}{\partial\bfn}&= {\bf0}
&&\textmd{on}\;\Gamma_{NT}, \label{eq7}
\\
\frac{\partial\theta}{\partial\bfn} &= 0 &&
\textmd{on}\;\Gamma_{NT}, \label{boundary temperature2}
\\
\bfu(0) &= \bfu_0 &&\textmd{in}\;\om, \label{eq8}
\\
\theta(0) &= \theta_0 &&\textmd{in}\;\om. \label{init_temp}
\end{align}
Here $g$ is a given function representing the distribution of the
temperature $\theta$ on $\Gamma_D$, $\bfu_0$ and  $\theta_0$
describe the initial velocity and temperature, respectively. Here we
suppose that all functions in \eqref{eq3}--\eqref{init_temp} are
smooth enough and satisfy the compatibility conditions $\bfu_0=\bf0$
and $\theta_0=g $ on $\Gamma_D$.
\begin{remark}
Throughout the paper, $\mu$ denotes some fixed (arbitrarily small)
positive real number (cf. \eqref{assump_f}).
\end{remark}
At this point we can formulate our problem. Suppose that
\begin{align}
& \bff \in L^{2+\mu}(0,T;\mathbf{V}_{u}^{0,2}), \quad h \in
L^2(0,T;V_{\theta}^{0,2}), \label{assump_f}
\\
& g\in { L^{2}}(0,T;W^{2,2}(\Omega)),\quad  g_t \in
L^2(0,T;L^2(\Omega)), \nonumber
\\
& \bfu_0 \in \mathcal{D}_u, \quad \theta_0 \in W^{2,2}(\Omega),\quad
\theta_0 - g(0) \in \mathcal{D}_{\theta}. \nonumber
\end{align}
Find a pair $[\bfu,\theta]$ such that
\begin{align}
& \bfu_t \in L^2( 0,T;\mathbf{V}_{u}^{0,2} ), \quad \bfu \in
L^2(0,T;\mathcal{D}_u) \cap L^{\infty}(0,T; \mathbf{V}_{u}^{1,2}),
\\
& \theta_t-g_t \in L^2( 0,T;V_{\theta}^{0,2} ), \quad  \theta - g
\in L^2(0,T;\mathcal{D}_{\theta}) \cap
L^{\infty}(0,T;V_{\theta}^{1,2})
\end{align}
and the following system
\begin{eqnarray}
(\bfu_t,\bfv) + a_u(\bfu,\bfv) + b(\bfu,\bfu,\bfv)+( \theta\bff,\bfv
)
 &=& ( \bff,\bfv ), \label{weak_NS}
 \\
(\theta_t,\varphi)_{\Omega}+ a_{\theta}(\theta,\varphi)+
d(\bfu,\theta,\varphi)-e(\bfu,\bfu,\varphi)
-(\theta\bff\cdot\bfu,\varphi)_{\Omega}&=&(h,\varphi)_{\Omega}\quad
\label{weak_heat}
\end{eqnarray}
holds for every $[\bfv,\varphi] \in \mathbf{V}_{u}^{1,2} \times
V_{\theta}^{1,2}$ and for almost every $t\in(0,T)$ and
\begin{align}
\qquad \qquad \qquad \bfu(0) & =  \bfu_0 && \textmd{ in } \Omega, &
\label{weak_init_vel}
\\
\theta(0) & =  \theta_0 && \textmd{ in } \Omega.  &
\label{weak_init_temp}
\end{align}
The pair $[\bfu,\theta]$ is called the strong solution to the system
\eqref{eq3}--\eqref{init_temp}.

Let us briefly describe some difficulties we have to solve in our
work. The equations  \eqref{eq3}--\eqref{heat equation} represent
the system with strong nonlinearities (quadratic growth of
$\nabla\bfu$ in dissipative term $\bfe(\bfu):\bfe(\bfu)$) without
appropriate general existence and regularity theory. In
\cite{frehse}, Frehse presented a simple example of discontinuous
bounded weak solution $\bfU\in L^{\infty}\cap H^1$ of nonlinear
elliptic system of the type $\Delta\bfU=B(\bfU,\nabla\bfU)$, where
$B$ is analytic and has quadratic growth in $\nabla\bfU$. However,
for scalar problems, such existence and regularity theory is well
developed (cf. \cite{LadUr,LadSolUr}).

Nevertheless, the main (open) problem of the system
\eqref{eq3}--\eqref{init_temp} consists in the fact that, because of
the boundary condition \eqref{eq7}, we cannot prove that
$b(\bfu,\bfu,\bfu)=0$. Consequently, we are not able to show that
the kinetic energy of the fluid is controlled by the data of the
problem and the solutions of \eqref{eq3}--\eqref{init_temp} need not
satisfy the energy inequality. This is due to the fact that some
uncontrolled ``backward flow'' can take place at the open parts
$\Gamma_N$ of the domain $\Omega$ and one is not able to prove
global (in time) existence results. In
\cite{Kracmar2002}--\cite{KraNeu5}, Kra\v cmar and Neustupa
prescribed an additional condition on the output (which bounds the
kinetic energy of the backward flow) and formulated steady and
evolutionary Navier--Stokes problems by means of appropriate
variational inequalities. In \cite{KuceraSkalak1998}, Ku\v cera and
Skal\' ak proved the local--in--time existence and uniqueness of a
``weak'' solution of the evolution Navier--Stokes problem for
iso-thermal fluids, such that
\begin{equation}\label{reg_kucera_skalak}
\bfu_t\in L^2(0,T^*; \mathbf{V}_{u}^{1,2}),\quad \bfu_{tt}\in
L^2(0,T^*;(\mathbf{V}_{u}^{1,2})^*), \quad 0<T^*\leq T,
\end{equation}
under some smoothness restrictions on $\bfu_0$ and $\mathcal{P}$. In
\cite{SkalakKucera2000}, the same authors established the similar
results for the evolution Boussinesq approximations of the
heat--conducting incompressible fluids. In \cite{Kucera2009},
Ku\v{c}era supposed that the ``do nothing'' problem for the
Navier--Stokes system is solvable in suitable function class with
some given data (the initial velocity and the right hand side). The
author proved that there exists a unique solution for data which are
small perturbations of the original ones.

In the present paper, we extend the results by Skal\' ak and Ku\v
cera in \cite{SkalakKucera2000}. However, our results are restricted
to two dimensions. We shall prove local existence and global
uniqueness of the strong solution to \eqref{eq3}--\eqref{init_temp},
i.e. for more general model than the Boussinesq approximations
considered in cited references, and, moreover, such that i.a.
$[\bfu,\theta]\in L^2(0,T^*;\mathbf{W}^{2,2})\times
L^2(0,T^*;{W}^{2,2})$, which is strong in the sense that the
solution possess second spatial derivatives. The main result of the
paper is the following
\begin{theorem}[Main result]\label{theorem_main_result}
Assume
\begin{align}
& \bff \in L^{2+\mu}(0,T;\mathbf{V}_{u}^{0,2}),  \quad h \in
L^2(0,T; V_{\theta}^{0,2}),
\\
& g\in {L^{2}}(0,T;W^{2,2}(\Omega)),\quad  g_t \in
L^2(0,T;L^2(\Omega)), \nonumber
\\
& \bfu_0 \in \mathcal{D}_u, \quad \theta_0 \in W^{2,2}(\Omega),\quad
\theta_0 - g(0) \in \mathcal{D}_{\theta}.\nonumber
\end{align}
Then there exists $T^*\in(0,T]$ and the pair $[\bfu,\theta]$,
\begin{align}
& \bfu_t \in L^2( 0,T;\mathbf{V}_{u}^{0,2} ), \quad \bfu \in
L^2(0,T;\mathcal{D}_u) \cap L^{\infty}(0,T; \mathbf{V}_{u}^{1,2}),
\\
& \theta_t-g_t \in L^2( 0,T; V_{\theta}^{0,2} ), \quad \theta-g \in
L^2(0,T;\mathcal{D}_{\theta}) \cap L^{\infty}(0,T;
V_{\theta}^{1,2}),
\end{align}
such that $[\bfu,\theta]$ is the strong solution of the system
\eqref{eq3}--\eqref{init_temp}. This solution is also globally
unique.
\end{theorem}

% ----------------------------------------------------------------
\section{Auxiliary results}\label{Auxiliary results}
Before we proceed to prove the main result of our paper, let us
establish the well-possedness property for appropriate linear
problems, which can be found in literature.
\begin{theorem}[The decoupled Stokes equations]\label{stokes_problem}
Let $\bff\in L^2(0,T;\mathbf{V}_{u}^{0,2})$. Then there exists the
unique function $\bfu\in L^2(0,T;\mathcal{D}_u) \cap
L^{\infty}(0,T;\mathbf{V}_{u}^{1,2})$, $\bfu_t\in
L^2(0,T;\mathbf{V}_{u}^{0,2})$, such that
\begin{equation} \label{lin_var_form_1a}
(\bfu_t,\bfv) + a_u(\bfu,\bfv) = (\bff,\bfv)
\end{equation}
holds for every $\bfv\in \mathbf{V}_{u}^{1,2}$ and for almost every
$t\in(0,T)$ and $\bfu(0) = \bfzero$. Moreover
\begin{equation}\label{eq11a}
\|\bfu_t\|_{L^2(0,T;\mathbf{V}_{u}^{0,2})}
+\|\bfu\|_{L^2(0,T;\mathcal{D}_{\theta})}
+\|\bfu\|_{L^{\infty}(0,T;\mathbf{V}_{u}^{1,2})} \leq c(\om)
\|\bff\|_{L^2(0,T;\mathbf{V}_{u}^{0,2})}.
\end{equation}
\end{theorem}
\begin{theorem}[The decoupled heat equation]\label{linear_heat_problem}
Let $h\in L^2(0,T;V_{\theta}^{0,2})$. Then there exists the unique
function $\theta\in L^2(0,T;\mathcal{D}_{\theta}) \cap
L^{\infty}(0,T;V_{\theta}^{1,2})$, $\theta_t\in
L^2(0,T;V_{\theta}^{0,2})$, such that
\begin{equation} \label{lin_var_form_2a}
(\theta_t,\varphi)_{\Omega} + a_{\theta}(\theta,\varphi) =
(h,\varphi)_{\Omega}
\end{equation}
holds for every $\varphi\in V_{\theta}^{1,2}$ and for almost every
$t\in(0,T)$ and $\theta(0) = 0$. Moreover
\begin{equation}\label{eq22a}
\|\theta_t\|_{L^2(0,T;V_{\theta}^{0,2})}
+\|\theta\|_{L^2(0,T;\mathcal{D}_{\theta})}
+\|\theta\|_{L^{\infty}(0,T;V_{\theta}^{1,2})} \leq c(\om)
\|h\|_{L^2(0,T;V_{\theta}^{0,2})}.
\end{equation}
\end{theorem}
\begin{proof}[Remark to proofs of Theorem
\ref{linear_heat_problem} and \ref{stokes_problem}]
The assertion of Theorem \ref{stokes_problem} is proved in
\cite[Theorem 2.1]{BeKuc2007} and \cite[Theorem 3.4]{BeKuc2008}. We
omit the proof of Theorem \ref{linear_heat_problem} since it can be
established in the same way, see also \cite[Chapter 5]{grisvard}.
Note that the well-known approach to the proof is based on the
Galerkin approximation with spectral basis and the uniform
boundedness of approximate solutions in suitable spaces, see also
\cite[Chapter 3]{Te}.
\end{proof}

%%%%%%%%%%%%%%%%%%%%%%%%%%%%%%%%%%%%%%%%%%%%%%%%%%%%%%%%%%%%%%%%%%%%%%%%%

\section{Proof of the main result}\label{sec proof_main}

\subsection{Existence}\label{proof_existence}

Let us introduce the following reflexive Banach spaces
\begin{eqnarray}
X^{u}_{T} &:=& \left\{ \bfphi \;|\; \bfphi\in
L^2(0,T;\mathcal{D}_u),\; \bfphi_t \in
L^2(0,T;\mathbf{V}_{u}^{0,2}),\; \bfphi(0)={\bf0} \right\},
\\
X^{\theta}_{T} &:=& \left\{ \psi \;|\; \psi\in
L^2(0,T;\mathcal{D}_{\theta}),\;  \psi_t \in
L^2(0,T;V_{\theta}^{0,2}),\; \psi(0)=0 \right\},
\\
X_{T} &:=& \left\{ [\bfphi,\psi]\;|\; \bfphi \in X^{u}_{T}, \;
\psi\in X^{\theta}_{T} \right\},
\end{eqnarray}
respectively, with norms
\begin{eqnarray}
\|\bfphi\|_{X^u_{T}} &:=& \dc{ \bfphi}\dc_{L^2(0,T;\,\mathcal{D}_u)}
+ \dc{ \bfphi}_t\dc_{L^2(0,T;\,\mathbf{V}_{u}^{0,2})},
\\
\|\psi\|_{X^{\theta}_{T}} &:=& \dc{
\psi}\dc_{L^2(0,T;\,\mathcal{D}_{\theta})} +
\dc{\psi}_t\dc_{L^2(0,T;\,V_{\theta}^{0,2})},
\\
\|[\bfphi,\psi]\|_{X_{T}} &:=& \dc{\bfphi}\dc_{X^u_{T}} + \dc{
\psi}\dc_{X^{\theta}_{T}}.
\end{eqnarray}
Let us present some properties of $X^u_T$, $X^{\theta}_T$ and
(consequently) $X_T$. Note that \cite[Remark 2.2 and Corollary
2.3]{Benes2009} yields the embedding $\mathcal{D}_u\hookrightarrow
\mathbf{W}^{2,2}$, which implies
\begin{equation}\label{emb_130}
X^u_{T} \hookrightarrow L^{\infty}(0,T;\mathbf{W}^{1,2})
\hookrightarrow L^{\infty}(0,T;\mathbf{L}^{p}), \quad 1\leq
p<\infty.
\end{equation}
Let $\bfphi \in X^u_{T}$.  Using the interpolation inequality
\begin{equation}\label{int_ineq_010}
\| \bfphi\|_{\mathbf{W}^{3/2,2}} \leq c \| \bfphi
\|^{1/2}_{\mathbf{W}^{1,2}} \| \bfphi\|^{1/2}_{\mathbf{W}^{2,2}}
\end{equation}
we get
\begin{eqnarray*}\nonumber
\|\bfphi\|_{L^4(0,T;\mathbf{W}^{3/2,2})} &\leq& c
\|\bfphi\|^{1/2}_{L^2(0,T;\mathbf{W}^{2,2})}
\|\bfphi\|^{1/2}_{L^\infty(0,T;\mathbf{W}^{1,2})}
\\
&\leq& \,c\,\|\bfphi\|_{X^u_{T}},
\end{eqnarray*}
where $c=c(\Omega)$. Hence we have
\begin{equation}\label{emb_160}
X^u_{T} \hookrightarrow L^{4}(0,T;\mathbf{W}^{3/2,2})\hookrightarrow
L^{4}(0,T;\mathbf{W}^{1,4})\hookrightarrow L^{4}(0,T;\mathbf{L}^{p})
\end{equation}
for every $1\leq p\leq\infty$. Similar properties hold for the space
$X^{\theta}_T$. Since \cite{KufSan} yields the embedding
$\mathcal{D}_{\theta}\hookrightarrow W^{2,2}(\Omega)$, we get
\begin{equation}\label{emb_230}
X^{\theta}_{T} \hookrightarrow L^{\infty}(0,T;{W}^{1,2}(\Omega))
\hookrightarrow L^{\infty}(0,T;{L}^{p}(\Omega)), \quad 1\leq
p<\infty,
\end{equation}
and
\begin{equation}\label{emb_260}
X^{\theta}_{T} \hookrightarrow
L^{4}(0,T;{W}^{3/2,2}(\Omega))\hookrightarrow
L^{4}(0,T;{W}^{1,4}(\Omega))\hookrightarrow
L^{4}(0,T;{L}^{p}(\Omega))
\end{equation}
for every $1\leq p\leq\infty$.

Let $\phi\in L^{2}(0,T;{W}^{2,2}(\Omega)) \cap
L^{\infty}(0,T;{W}^{1,2}(\Omega))$. Raising and integrating the
interpolation inequality
\begin{equation*}
\| \phi(t) \|_{W^{1+j,2}(\om)} \leq c\| \phi(t)
\|^{j}_{W^{2,2}(\om)} \|\phi(t)\|^{1-j}_{W^{1,2}(\om)}, \; 0 < j
\leq 1,
\end{equation*}
from $0$ to $T$ we get
\begin{eqnarray}\label{est_121}
 \left( \int^T_0 \| \phi(t) \|^{2/j}_{W^{1+j,2}(\om)}
 {\rm d}t \right)^{j/2}  &\leq& c\left( \int^T_0 \| \phi(t)
\|^2_{W^{2,2}(\om)}\|\phi(t)\|^{2(1-j)/j}_{W^{1,2}(\om)} {\rm d}t
\right)^{j/2}
\nonumber\\
&\leq& c\| \phi \|^{j}_{{L}^{2}(0,T;W^{2,2}(\om))} \| \phi
\|^{1-j}_{L^{\infty}(0,T;W^{1,2}(\om))},
\end{eqnarray}
where $c=c(\Omega)$. Hence we have
\begin{align}
L^{2}(0,T;{W}^{2,2}(\Omega)) \cap L^{\infty}(0,T;{W}^{1,2}(\Omega))
&\hookrightarrow L^{2/j}(0,T;{W}^{1+j,2}(\Omega))
\\
&\hookrightarrow L^{2/j}(0,T;{L}^{\infty}(\Omega))
\end{align}
for arbitrarily small positive $j$ and consequently we deduce
\begin{equation}
X^{u}_{T} \hookrightarrow L^{q}(0,T;\mathbf{L}^{\infty})
\end{equation}
and
\begin{equation}
X^{\theta}_{T} \hookrightarrow L^{q}(0,T;{L}^{\infty}(\Omega))
\end{equation}
for $2\leq q<\infty$.

%----------------------------------------------------------------------

\begin{remark}\label{subst_hom_problem}
Setting $[\bfw,\vartheta]=[\bfu-\unu,\theta-\theta_0-g]$ this
amounts to solving the problem with homogeneous boundary and initial
data
\begin{eqnarray}
(\bfw_t,\bfv)+a_u(\bfw+\unu,\bfv)+b(\bfw+\unu,\bfw+\unu,\bfv) &&
\nonumber
\\
-(\bff (\vartheta+\theta_0+g),\bfv) &=& (\bff,\bfv), \nonumber
\\
\label{var_form_2a}
\\
(\vartheta_t+g_t,\varphi)_{\Omega}+
a_{\theta}(\vartheta+\theta_0+g,\varphi)\nonumber +
d(\bfw+\unu,\vartheta+\theta_0+g,\varphi)
\\
+((\vartheta+\theta_0+g)\bff\cdot(\bfw+\unu),\varphi)_{\Omega}
-e(\bfw+\unu,\bfw+\unu,\varphi) &=& ( h,\varphi )_{\Omega},
\nonumber\\
\\
\bfw(0) &=& {\bfzero}, \label{var_form_2b}
\\
\vartheta(0)&=&0
\end{eqnarray}
for every $[\bfv,\varphi]\in \mathbf{V}_{u}^{1,2} \times
V_{\theta}^{1,2}$ and for almost every $t\in(0,T)$, where $\bff\in
L^2(0,T;\mathbf{V}_{u}^{0,2})$, $h \in L^2(0,T;{V}_{\theta}^{0,2})$,
$\unu\in \mathcal{D}_u$ and $\theta_0 \in \mathcal{D}_{\theta}$.
\end{remark}
For arbitrary fixed
$[\widetilde{\bfw},\widetilde{\vartheta}]\in{X}_{T}$  we now
consider the nonlinear problem
\begin{eqnarray}
(\bfw_t,\bfv) + a_u(\bfw,\bfv)  &=& (\bff,\bfv)-a_u(\unu,\bfv) -
b(\unu,\unu,\bfv)\nonumber
\\
&&-b(\unu,\widetilde{\bfw},\bfv)-b(\widetilde{\bfw},\unu,\bfv)
-b(\widetilde{\bfw},\widetilde{\bfw},\bfv) \nonumber
\\
&&+(\bff (\theta_0+g),\bfv) +(\bff \widetilde{\vartheta},\bfv),
\nonumber
\\
\label{linear_problem_1}
\\
(\vartheta_t,\varphi)_{\Omega}+a_{\theta}(\vartheta,\varphi)+e(\bfw,\bfw,\varphi)&=&
-2e(\widetilde{\bfw},\unu,\varphi)-e(\unu,\unu,\varphi)+(h,\varphi)_{\Omega}
\nonumber
\\
&&-
(g_t,\varphi)_{\Omega}-a_{\theta}(\theta_0+g,\varphi)-d(\unu,\theta_0+g,\varphi)
\nonumber
\\
&&-d(\widetilde{\bfw},\widetilde{\vartheta},\varphi)
-d(\widetilde{\bfw},\theta_0+g,\varphi)
-d(\unu,\widetilde{\vartheta},\varphi) \nonumber
\\
&& + ((\theta_0+g)\bff\cdot\unu,\varphi)_{\Omega} +
(\widetilde{\vartheta}\bff\cdot\unu,\varphi)_{\Omega}
\nonumber\\
&& + ((\theta_0+g)\bff\cdot\widetilde{\bfw},\varphi)_{\Omega}+
(\widetilde{\vartheta}\bff\cdot\widetilde{\bfw},\varphi)_{\Omega},
\nonumber\\
\label{linear_problem_2}
\\
\bfw(0) &=& \bfzero, \label{linear_problem_3}
\\
\vartheta(0)&=&0 \label{linear_problem_4}
\end{eqnarray}
for every $[\bfv,\varphi]\in \mathbf{V}_{u}^{1,2} \times
V_{\theta}^{1,2}$ and for almost every $t\in(0,T)$.

In the proof of the next lemma we verify that all terms on the
right-hand sides of
\eqref{linear_problem_1}--\eqref{linear_problem_2} and the
dissipative term $e(\bfw,\bfw,\cdot)$ are well defined. From Theorem
\ref{stokes_problem} and Theorem \ref{linear_heat_problem} we deduce
that for arbitrary fixed
$[\widetilde{\bfw},\widetilde{\vartheta}]\in{X}_{T}$ there exists
$[\bfw,\vartheta]=\mathcal{Z}([\widetilde{\bfw},\widetilde{\vartheta}])\in{X}_{T}$,
the solution of the problem
\eqref{linear_problem_1}--\eqref{linear_problem_4}. Consequently,
the mapping $\mathcal{Z}:{X}_{T}\rightarrow{X}_{T}$ is well defined.

Denote by $B_R(T) \subset X_{T}$ the closed ball
\begin{equation}
B_R(T):=\{[\bfphi,\psi]\in X_{T};\, \|[\bfphi,\psi]\|_{{X}_{T}}\leq
R\}.
\end{equation}
We are going to show that the formal map
$\mathcal{Z}:[\widetilde{\bfw},\widetilde{\vartheta}]\rightarrow
[\bfw,\vartheta]$, where $[\bfw,\vartheta]$ is a solution of the
problem \eqref{linear_problem_1}--\eqref{linear_problem_4}, has a
fixed point in $B_R(T^*)$ for $T^*$ sufficiently small.

Let us prepare the existence proof. We now claim the following
\begin{lemma}\label{ball_into}
For all $R>0$, $T>0$ and for every
$[\widetilde{\bfw},\widetilde{\vartheta}]\in B_R(T)$ we have
\begin{equation}\label{est_Z}
\|\mathcal{Z}([\widetilde{\bfw},\widetilde{\vartheta}])\|_{X_T}\leq
c_1\xi_1(T)+c_2\xi_2(T)\omega(R),
\end{equation}
where $\omega$ is an increasing function depending solely on $R$,
 functions $\xi_1$ and $\xi_2$ depend solely on $T$ and
$\xi_1(T)\rightarrow 0_+$, $\xi_2(T)\rightarrow 0_+$ for
$T\rightarrow 0_+$ and the positive constants $c_1$ and $c_2$ are
independent of $T$ and $R$.
\end{lemma}
\begin{proof}
The proof is rather technical. Obviously, there exists a positive
function $K_1(T)$ such that
\begin{multline}\label{eq37a}
\|(\bff,\cdot) - a_u(\unu,\cdot) - b(\unu,\unu,\cdot)+( \bff
(\theta_0+g),\cdot)\|_{L^2(0,T;\mathbf{V}_{u}^{0,2})}
\\
+ \|(h,\cdot)_{\Omega}-
(g_t,\cdot)_{\Omega}-a_{\theta}(\theta_0+g,\cdot)-
d(\unu,\theta_0+g,\cdot)\|_{L^2(0,T;V_{\theta}^{0,2})}
\\
+ \|e(\unu,\unu,\cdot)+((\theta_0+g)\bff\cdot\unu,\cdot)_{\Omega}
\|_{L^2(0,T;V_{\theta}^{0,2})} \leq c K_1(T)
\end{multline}
and $K_1(T)\rightarrow 0_+$ for $T\rightarrow 0_+$, $c$ is
independent of $T$.

Similarly and using the interpolation inequality we get
\begin{eqnarray}\label{eq38}
\|
b(\widetilde{\bfw},\unu,\cdot)\|_{L^2(0,T;\mathbf{V}_{u}^{0,2})}&\leq&
\left(\int_0^T\|\widetilde{\bfw}\|_{\mathbf{L}^4}^2\|\nabla\unu\|^2_{\mathbf{L}^{4}}
{\rm d}t\right)^{1/2} \nonumber
\\
\nonumber &\leq&
\|\unu\|_{\mathbf{W}^{1,4}}\|\widetilde{\bfw}\|_{L^2(0,T;\mathbf{L}^{4})}
\\
\nonumber &\leq&  T^{1/4} \|\unu\|_{\mathbf{W}^{1,4}}
\|\widetilde{\bfw}\|_{L^{4}(0,T;\mathbf{L}^{4})}
\\
&\leq& c
T^{1/4}\|\unu\|_{\mathcal{D}_u}\|\widetilde{\bfw}\|_{{X}^u_{T}}.
\end{eqnarray}
The other terms in \eqref{linear_problem_1} can be handled in the
same way
\begin{eqnarray}\label{eq39}
\|
b(\unu,\widetilde{\bfw},\cdot)\|_{L^2(0,T;\mathbf{V}_{u}^{0,2})}&\leq&
\left(\int_0^T\|\unu\|^2_{\mathbf{L}^{\infty}}
\|\nabla\widetilde{\bfw}\|_{\mathbf{L}^{2}}^2 {\rm d}t\right)^{1/2}
\nonumber \\
&\leq& \|\unu\|_{\mathbf{L}^{\infty}}\|\widetilde{\bfw}
\|_{L^2(0,T;\mathbf{W}^{1,2})}
\nonumber \\
 &\leq&  T^{1/4} \|\unu\|_{\mathbf{L}^{\infty}}
\|\widetilde{\bfw}\|_{L^{4}(0,T;\mathbf{W}^{1,2})}
\nonumber\\
&\leq& c \,
T^{1/4}\|\unu\|_{\mathcal{D}_u}\|\widetilde{\bfw}\|_{{X}^u_{T}},
\end{eqnarray}
\begin{eqnarray}\label{eq40}
\|b(\widetilde{\bfw},\widetilde{\bfw},\cdot)\|_{L^2(0,T;\mathbf{V}_{u}^{0,2})}
&\leq& \left(\int_0^T \|\widetilde{\bfw}\|^2_{\mathbf{L}^{4}}
\|\nabla\widetilde{\bfw}\|^2_{\mathbf{L}^{4}}{\rm d}t\right)^{1/2}
\nonumber\\
&\leq&  T^{1/4}
\|\widetilde{\bfw}\|_{L^{\infty}(0,T;\mathbf{L}^{4})}
\|\widetilde{\bfw}\|_{L^{4}(0,T;\mathbf{W}^{1,4})}
\nonumber  \\
 &\leq& c \, T^{1/4}\|\widetilde{\bfw}\|^2_{{X}^u_{T}}
\end{eqnarray}
and finally the last term in \eqref{linear_problem_1} can be
estimated using the interpolation inequality as
\begin{eqnarray}\label{eq41}
\|(\bff
\widetilde{\vartheta},\cdot)\|_{L^2(0,T;\mathbf{V}_{u}^{0,2})}
&\leq& \left(\int_0^T \|\bff\|^{2}_{\mathbf{L}^{2}}
\|\widetilde{\vartheta}\|^2_{L^{\infty}(\om)}
 {\rm d}t\right)^{1/2}
\nonumber\\
&\leq&
T^{\mu/[(2+\mu)(4+\mu)]}\|\bff\|_{L^{2+\mu}(0,T;\mathbf{L}^{2})}
\|\widetilde{\vartheta}\|_{L^{2(4+\mu)/\mu}(0,T;L^{\infty}(\om))}
\nonumber\\
&\leq& c
T^{\mu/[(2+\mu)(4+\mu)]}\|\bff\|_{L^{2+\mu}(0,T;\mathbf{L}^{2})}
\|\widetilde{\vartheta}\|_{{X}^{\theta}_{T}}.
\nonumber\\
\end{eqnarray}
The inequalities \eqref{eq37a}--\eqref{eq41} and Theorem
\ref{stokes_problem} yield the estimate
\begin{equation}\label{eq60}
\|\bfw\|_{{X}^{u}_{T}} \leq K_1(T) + c(R^2+R) T^{1/4}+ c
T^{\mu/[(2+\mu)(4+\mu)]}R,
\end{equation}
where $K_1(T)\rightarrow 0_+$ for $T\rightarrow 0_+$ and $c$ is
independent of $T$. Successively, we use \eqref{emb_160} and
\eqref{eq60} to estimate the dissipative term $e(\bfw,\bfw,\cdot)$
in \eqref{linear_problem_2}
\begin{eqnarray}\label{}
\|e(\bfw,\bfw,\cdot)\|_{L^2(0,T;{V}_{\theta}^{0,2})} &\leq&
\left(\int_0^T \|\bfe({\bfw})\|^4_{\mathbf{L}^{4}} {\rm
d}t\right)^{1/2}
\nonumber\\
&\leq& \|{\bfw}\|^2_{L^{4}(0,T;\mathbf{W}^{1,4})}
\nonumber\\
&\leq&c \|\bfw\|^2_{{X}^{u}_{T}}
\nonumber\\
&\leq& \left[K_1(T) + c(R^2+R) T^{1/4}+ c
T^{\mu/[(2+\mu)(4+\mu)]}R\right]^2.
\nonumber\\
\end{eqnarray}
Now we are going to estimate all remaining terms on the right-hand
side of \eqref{linear_problem_2}.  We use interpolation inequality
and \eqref{emb_160}
 to estimate the dissipative term $e(\widetilde{\bfw},\unu,\varphi)$
\begin{eqnarray}\label{eq45}
\|e(\widetilde{\bfw},\unu,\cdot)\|_{L^2(0,T;{V}_{\theta}^{0,2})}
&\leq&\left(\int_0^T\|\nabla\unu\|^2_{\mathbf{L}^{4}}
\|\nabla\widetilde{\bfw}\|^2_{\mathbf{L}^{4}} {\rm d}t\right)^{1/2}
\nonumber\\
 &\leq& T^{1/4} \|\nabla\unu\|_{\mathbf{L}^{4}}
\|\nabla\widetilde{\bfw}\|_{L^{4}(0,T;\mathbf{L}^4)}
\nonumber\\
&\leq& c \,
T^{1/4}\|\unu\|_{\mathcal{D}_u}\|\widetilde{\bfw}\|_{{X}^u_{T}}
\end{eqnarray}
and taking into account \eqref{emb_160} and \eqref{emb_260} we
arrive at the estimate
\begin{eqnarray}\label{eq42}
\|d(\widetilde{\bfw},\widetilde{\vartheta},\cdot)\|_{L^2(0,T;{V}_{\theta}^{0,2})}
&\leq& \left(\int_0^T \|\widetilde{\bfw}\|^2_{\mathbf{L}^4}
\|\nabla\widetilde{\vartheta}\|^2_{\mathbf{L}^{4}(\om)}{\rm
d}t\right)^{1/2}
\nonumber\\
&\leq&  T^{1/4}
\|\widetilde{\bfw}\|_{L^{\infty}(0,T;\mathbf{L}^{4})}
\|\widetilde{\vartheta}\|_{L^{4}(0,T;W^{1,4}(\om))}
\nonumber\\
 &\leq& c
T^{1/4}\|\widetilde{\bfw}\|_{{X}^{u}_{T}}
\|\widetilde{\vartheta}\|_{{X}^{\theta}_{T}}
\nonumber\\
&\leq&
cT^{1/4}\frac{1}{2}\left(\|\widetilde{\bfw}\|^2_{{X}^{u}_{T}}+
\|\widetilde{\vartheta}\|^2_{{X}^{\theta}_{T}}\right).
\end{eqnarray}
Let us proceed to estimate the remaining convective terms:
\begin{eqnarray}\label{eq43}
\|d(\widetilde{\bfw},\theta_0+g,\cdot)\|_{L^2(0,T;{V}_{\theta}^{0,2})}
&\leq& \left(\int_0^T \|\widetilde{\bfw}\|^2_{\mathbf{L}^{4}}
\|\theta_0+g\|^2_{W^{1,4}(\om)}{\rm d}t\right)^{1/2}
\nonumber\\
&\leq&  T^{1/4}
\|\widetilde{\bfw}\|_{L^{\infty}(0,T;\mathbf{L}^{4})}
\|\theta_0+g\|_{L^{4}(0,T;W^{1,4}(\om))}
\nonumber\\
 &\leq& c
T^{1/4}\|\widetilde{\bfw}\|_{{X}^{u}_{T}}\|\theta_0+g\|_{{X}^{\theta}_{T}}
\nonumber\\
&\leq&
cT^{1/4}\frac{1}{2}\left(\|\widetilde{\bfw}\|^2_{{X}^{u}_{T}}+
\|\theta_0+g\|^2_{{X}^{\theta}_{T}}\right)
\nonumber\\
&\leq& c T^{1/4}\frac{1}{2}\|\widetilde{\bfw}\|^2_{{X}^{u}_{T}} +
K_2(T),
\end{eqnarray}
where the positive function $K_2(T)\rightarrow 0_+$ for
$T\rightarrow 0_+$. Further
\begin{eqnarray}\label{eq44}
\|d(\unu,\widetilde{\vartheta},\cdot)\|_{L^2(0,T;{V}_{\theta}^{0,2})}
&\leq& \left(\int_0^T \|\unu\|^2_{\mathbf{L}^{\infty}}
\|\nabla\widetilde{\vartheta}\|^2_{\mathbf{L}^{2}}{\rm
d}t\right)^{1/2}
\nonumber\\
&\leq&  T^{1/4} \|\unu\|_{\mathbf{L}^{\infty}}
\|\widetilde{\vartheta}\|_{L^{4}(0,T;W^{1,2}(\om))}
\nonumber\\
 &\leq& c
T^{1/4}\|\unu\|_{\mathcal{D}_u}\|\widetilde{\vartheta}\|_{{X}^{\theta}_{T}}.
\end{eqnarray}
In the similar way one can deduce for the adiabatic terms the
estimates
\begin{multline}\label{eq46}
\| (\widetilde{\vartheta}\bff\cdot\widetilde{\bfw},\cdot)_{\Omega}
\|_{L^2(0,T;{V}_{\theta}^{0,2})} \leq \left(\int_0^T
\|\widetilde{\vartheta}\|^2_{L^{\infty}(\om)}
\|\bff\|^2_{\mathbf{L}^{2}}
\|\widetilde{\bfw}\|^2_{\mathbf{L}^{\infty}} {\rm d}t\right)^{1/2}
\\
 \leq
 \|\widetilde{\vartheta}\|_{L^{4(4+\mu)/\mu}(0,T;L^{\infty}(\om))}
 \|\bff\|_{L^{2+\mu/2}(0,T;\mathbf{L}^{2})}
 \|\widetilde{\bfw}\|_{L^{4(4+\mu)/\mu}(0,T;\mathbf{L}^{\infty})}
\\
 \leq c
 \|\widetilde{\vartheta}\|_{{X}^{\theta}_{T}}
 \|\bff\|_{L^{2+\mu/2}(0,T;\mathbf{L}^{2})}
 \|\widetilde{\bfw}\|_{{X}^{u}_{T}}
 \\
 \leq c T^{\mu/[(2+\mu)(4+\mu)]}
 \|\widetilde{\vartheta}\|_{{X}^{\theta}_{T}}
 \|\bff\|_{L^{2+\mu}(0,T;\mathbf{L}^{2})}
 \|\widetilde{\bfw}\|_{{X}^{u}_{T}}
\end{multline}
and, in the same way,
\begin{multline}\label{eq47}
\|
((\theta_0+g)\bff\cdot\widetilde{\bfw},\cdot)_{\Omega}\|_{L^2(0,T;{V}_{\theta}^{0,2})}
\\
 \leq
 \|\theta_0+g\|_{L^{4(4+\mu)/\mu}(0,T;L^{\infty}(\om))}
 \|\bff\|_{L^{2+\mu/2}(0,T;\mathbf{L}^{2})}
 \|\widetilde{\bfw}\|_{L^{4(4+\mu)/\mu}(0,T;\mathbf{L}^{\infty})}
\\
 \leq c
 \|\theta_0+g\|_{{X}^{\theta}_{T}}
 \|\bff\|_{L^{2+\mu/2}(0,T;\mathbf{L}^{2}(\om))}
 \|\widetilde{\bfw}\|_{{X}^{u}_{T}}
 \\
 \leq c T^{\mu/[(2+\mu)(4+\mu)]}
 \|\theta_0+g\|_{{X}^{\theta}_{T}}
 \|\bff\|_{L^{2+\mu}(0,T;\mathbf{L}^{2})}
 \|\widetilde{\bfw}\|_{{X}^{u}_{T}}.
\end{multline}
Finally,
\begin{multline}\label{eq48}
\|
(\widetilde{\vartheta}\bff\cdot\unu,\cdot)_{\Omega}\|_{L^2(0,T;{V}_{\theta}^{0,2})}
\leq \left(\int_0^T
\|\widetilde{\vartheta}\|^2_{L^{\infty}(\om)}\|\bff\|^{2}_{\mathbf{L}^{2}}
\|\unu\|^2_{\mathbf{L}^{\infty}} {\rm d}t\right)^{1/2}
\\
\leq T^{\mu/[(2+\mu)(4+\mu)]}
\|\widetilde{\vartheta}\|_{L^{2(4+\mu)/\mu}(0,T;L^{\infty}(\om))}
\|\bff\|_{L^{2+\mu}(0,T;\mathbf{L}^{2})}
\|\unu\|_{\mathbf{L}^{\infty}}
\\
\leq c T^{\mu/[(2+\mu)(4+\mu)]}
\|\widetilde{\vartheta}\|_{{X}^{\theta}_{T}}
\|\bff\|_{L^{2+\mu}(0,T;\mathbf{L}^{2})}
\|\unu\|_{\mathbf{L}^{\infty}}.
\end{multline}
Now the assertion is the straight consequence of inequalities
\eqref{eq37a}--\eqref{eq48}, Theorem \ref{linear_heat_problem} and
Theorem \ref{stokes_problem}. The proof of Lemma \ref{ball_into} is
complete.
\end{proof}
\begin{lemma}\label{ball_contraction}
For $T^* \in (0,T]$, sufficiently small, $\mathcal{Z}$ maps
$B_R(T^*)$ into $B_R(T^*)$ and realizes a contraction.
\end{lemma}
\begin{corollary}
Lemma \ref{ball_into}, Lemma \ref{ball_contraction} and the Banach
fixed point theorem yield the existence of a fixed-point
$[\bfw,\vartheta]=\mathcal{Z}([\bfw,\vartheta])$ in the ball
$B_R(T^*)$ for sufficiently small $T^* \in (0,T]$. By Remark
\ref{subst_hom_problem} the function
$[\bfu,\theta]=[\bfw+\unu,\vartheta+\theta_0+g]$ is the strong
solution to the system \eqref{eq3}--\eqref{init_temp} on $(0,T^*)$.
\end{corollary}
\begin{proof}[Proof of Lemma \ref{ball_contraction}]
Fix $R>0$ and choose $\tau\in(0,T]$ sufficiently small, such that
(cf. \eqref{est_Z})
$$
c_1\xi_1(s)+c_2\xi_2(s)\omega(R)\leq R \quad \forall s \in (0,\tau].
$$
Now we conclude that for every positive $T^*$, $0<T^*\leq \tau$,
$\mathcal{Z}$ maps $B_R(T^*)$ into $B_R(T^*)$.

Now we prove contraction. Let
$[\widetilde{\bfw}_1,\widetilde{\vartheta}_1],
[\widetilde{\bfw}_2,\widetilde{\vartheta}_2]\in X_{T}$, ${\bfz}=
\widetilde{\bfw}_2 - \widetilde{\bfw}_1$ and
$\sigma=\widetilde{\vartheta}_2-\widetilde{\vartheta}_1$. Theorem
\ref{linear_heat_problem} and Theorem \ref{stokes_problem} imply the
estimate
\begin{multline}\label{eq_contr_01}
\|\mathcal{Z}([\widetilde{\bfw}_2,\widetilde{\vartheta}_2])
-\mathcal{Z}([\widetilde{\bfw}_1,\widetilde{\vartheta}_1])\|_{X_T}
\leq c(\Omega) \left(
\|b(\bfz,\unu,\cdot)\|_{L^2(0,T;\mathbf{V}_{u}^{0,2})} \right.
\\
+\|b(\unu,\bfz,\cdot)\|_{L^2(0,T;\mathbf{V}_{u}^{0,2})}
+\|b(\widetilde{\bfw}_2,\bfz,\cdot)\|_{L^2(0,T;\mathbf{V}_{u}^{0,2})}
+\|b(\bfz,\widetilde{\bfw}_1,\cdot)\|_{L^2(0,T;\mathbf{V}_{u}^{0,2})}
\\
+\|(\bff \sigma,\cdot)\|_{L^2(0,T;\mathbf{V}_{u}^{0,2})}
+\|d(\bfz,\widetilde{\vartheta}_2,\cdot)\|_{L^2(0,T;{V}_{\theta}^{0,2})}
+\|d(\widetilde{\bfw}_1,\sigma,\cdot)\|_{L^2(0,T;{V}_{\theta}^{0,2})}
\\
+\|d(\bfz,\theta_0+g,\cdot)\|_{L^2(0,T;{V}_{\theta}^{0,2})}
+\|d(\unu,\sigma,\cdot)\|_{L^2(0,T;{V}_{\theta}^{0,2})}
\\
+\|e(\bfw_2-\bfw_1,\unu,\cdot)\|_{L^2(0,T;{V}_{\theta}^{0,2})}
\|e(\bfw_2-\bfw_1,\bfw_2,\cdot)\|_{L^2(0,T;{V}_{\theta}^{0,2})}
\\
\|e(\bfw_1,\bfw_2-\bfw_1,\cdot)\|_{L^2(0,T;{V}_{\theta}^{0,2})}
+\|(\sigma\bff\cdot\unu,\cdot)_{\Omega}\|_{L^2(0,T;{V}_{\theta}^{0,2})}
\\
+\|((\theta_0+g)\bff\cdot\bfz,\cdot)_{\Omega}\|_{L^2(0,T;{V}_{\theta}^{0,2})}
+\|(\sigma\bff\cdot\widetilde{\bfw}_1,\cdot)_{\Omega}\|_{L^2(0,T;{V}_{\theta}^{0,2})}
\\
\left.   \|
(\widetilde{\vartheta}_2\bff\cdot\bfz,\cdot)_{\Omega}\|_{L^2(0,T;{V}_{\theta}^{0,2})}
\right).
\end{multline}
Following \eqref{eq37a}--\eqref{eq44} we have (recall that
$c=c(\Omega)$)
\begin{eqnarray}
\|b(\bfz,\unu,\cdot)\|_{L^2(0,T;\mathbf{V}_{u}^{0,2})} &\leq& c \,
T^{1/4}\|\unu\|_{\mathcal{D}_u}\|\bfz\|_{{X}^u_{T}},
\label{eq_contr_10}
\\
\|b(\unu,\bfz,\cdot)\|_{L^2(0,T;\mathbf{V}_{u}^{0,2})} &\leq& c \,
T^{1/4}\|\unu\|_{\mathcal{D}_u}\|\bfz\|_{{X}^u_{T}},
\label{eq_contr_11}
\\
\|b(\widetilde{\bfw}_2,\bfz,\cdot)\|_{L^2(0,T;\mathbf{V}_{u}^{0,2})}
&\leq& c \,
T^{1/4}\|\widetilde{\bfw}_2\|_{{X}^u_{T}}\|\bfz\|_{{X}^u_{T}},
\label{eq_contr_12}
\\
\|b(\bfz,\widetilde{\bfw}_1,\cdot)\|_{L^2(0,T;\mathbf{V}_{u}^{0,2})}
&\leq& c \,
T^{1/4}\|\bfz\|_{{X}^u_{T}}\|\widetilde{\bfw}_1\|_{{X}^u_{T}},
\label{eq_contr_13}
\\
\|(\bff \sigma,\cdot)\|_{L^2(0,T;\mathbf{V}_{u}^{0,2})} &\leq&  c
T^{\mu/[(2+\mu)(4+\mu)]}\|\bff\|_{L^{2+\mu}(0,T;\mathbf{L}^{2})},
\|\sigma\|_{{X}^{\theta}_{T}}, \label{eq_contr_14}
\\
\|d(\bfz,\widetilde{\vartheta}_2,\cdot)\|_{L^2(0,T;{V}_{\theta}^{0,2})}
&\leq& c \,
T^{1/4}\|\bfz\|_{{X}^{u}_{T}}\|\widetilde{\vartheta}_2\|_{{X}^{\theta}_{T}},
\label{eq_contr_15}
\\
\|d(\widetilde{\bfw}_1,\sigma,\cdot)\|_{L^2(0,T;{V}_{\theta}^{0,2})}
&\leq& c \,
T^{1/4}\|\widetilde{\bfw}_1\|_{{X}^{u}_{T}}\|\sigma\|_{{X}^{\theta}_{T}},
\label{eq_contr_16}
\\
\|d(\bfz,\theta_0+g,\cdot)\|_{L^2(0,T;{V}_{\theta}^{0,2})} &\leq& c
\, T^{1/4}\|\bfz\|_{{X}^{u}_{T}}\|\theta_0+g\|_{{X}^{\theta}_{T}},
\label{eq_contr_17}
\\
\|d(\unu,\sigma,\cdot)\|_{L^2(0,T;{V}_{\theta}^{0,2})} &\leq&  c \,
T^{1/4}\|\unu\|_{\mathcal{D}_u}\|\sigma\|_{{X}^{\theta}_{T}}
\label{eq_contr_18} .
\end{eqnarray}
Estimating the dissipative terms we arrive at
\begin{eqnarray}
\|e(\bfw_2-\bfw_1,\unu,\cdot)\|_{L^2(0,T;{V}_{\theta}^{0,2})} &\leq&
c \,\|\bfw_2-\bfw_1\|_{{X}^u_{T}}\|\unu\|_{\mathcal{D}_u},
\label{eq_contr_19}
\\
\|e(\bfw_2-\bfw_1,\bfw_2,\cdot)\|_{L^2(0,T;{V}_{\theta}^{0,2})}
&\leq&  c \,\|\bfw_2-\bfw_1\|_{{X}^u_{T}}\|\bfw_2\|_{{X}^u_{T}},
\label{eq_contr_20}
\\
\|e(\bfw_1,\bfw_2-\bfw_1,\cdot)\|_{L^2(0,T;{V}_{\theta}^{0,2})}
&\leq&  c \,\|\bfw_1\|_{{X}^u_{T}}\|\bfw_2-\bfw_1\|_{{X}^u_{T}},
\label{eq_contr_21}.
\end{eqnarray}
Following Theorem \ref{stokes_problem} we can eliminate the term
$\|\bfw_2-\bfw_1\|_{{X}^u_{T}}$ in
\eqref{eq_contr_19}--\eqref{eq_contr_21} using the energy-like
estimate (cf. \eqref{eq11a})
\begin{multline}
\|\bfw_2-\bfw_1\|_{{X}^u_{T}} \leq c(\Omega) \left(
\|b(\bfz,\unu,\cdot)\|_{L^2(0,T;\mathbf{V}_{u}^{0,2})}
+\|b(\unu,\bfz,\cdot)\|_{L^2(0,T;\mathbf{V}_{u}^{0,2})} \right.
\\\left.
+\|b(\widetilde{\bfw}_2,\bfz,\cdot)\|_{L^2(0,T;\mathbf{V}_{u}^{0,2})}
+\|b(\bfz,\widetilde{\bfw}_1,\cdot)\|_{L^2(0,T;\mathbf{V}_{u}^{0,2})}
+\|(\bff \sigma,\cdot)\|_{L^2(0,T;\mathbf{V}_{u}^{0,2})}
 \right)
\end{multline}
and then apply the inequalities
\eqref{eq_contr_10}--\eqref{eq_contr_14}. Finally, the adiabatic
terms can be estimated in the following way:
\begin{equation}\label{eq_contr_22}
\|(\sigma\bff\cdot\unu,\cdot)_{\Omega}\|_{L^2(0,T;{V}_{\theta}^{0,2})}
\leq c T^{\mu/[(2+\mu)(4+\mu)]}
 \|\sigma\|_{{X}^{\theta}_{T}}
 \|\bff\|_{L^{2+\mu}(0,T;\mathbf{L}^{2})}
 \|\unu\|_{{X}^{u}_{T}},
\end{equation}
\begin{multline}\label{eq_contr_23}
\|((\theta_0+g)\bff\cdot\bfz,\cdot)_{\Omega}\|_{L^2(0,T;{V}_{\theta}^{0,2})}
\\
\leq c T^{\mu/[(2+\mu)(4+\mu)]}
 \|\theta_0+g\|_{{X}^{\theta}_{T}}
 \|\bff\|_{L^{2+\mu}(0,T;\mathbf{L}^{2})}
 \|\bfz\|_{{X}^{u}_{T}},
\end{multline}
\begin{multline}\label{eq_contr_24}
\|
(\sigma\bff\cdot\widetilde{\bfw}_1,\cdot)_{\Omega}\|_{L^2(0,T;{V}_{\theta}^{0,2})}
\\
\leq c T^{\mu/[(2+\mu)(4+\mu)]}
 \|\sigma\|_{{X}^{\theta}_{T}}
 \|\bff\|_{L^{2+\mu}(0,T;\mathbf{L}^{2})}
 \|\widetilde{\bfw}_1\|_{{X}^{u}_{T}}
\end{multline}
and
\begin{multline}\label{eq_contr_25}
\|
(\widetilde{\vartheta}_2\bff\cdot\bfz,\cdot)_{\Omega}\|_{L^2(0,T;{V}_{\theta}^{0,2})}
\\
\leq c T^{\mu/[(2+\mu)(4+\mu)]}
 \|\widetilde{\vartheta}_2\|_{{X}^{\theta}_{T}}
 \|\bff\|_{L^{2+\mu}(0,T;\mathbf{L}^{2})}
 \|\bfz\|_{{X}^{u}_{T}}  .
\end{multline}
Now let $[\widetilde{\bfw}_1,\widetilde{\vartheta}_1],
[\widetilde{\bfw}_2,\widetilde{\vartheta}_2]\in B_R(T) \subset
X_{T}$. Inequality \eqref{eq_contr_01} and the estimates
\eqref{eq_contr_10}--\eqref{eq_contr_25} yield
\begin{equation}
\|\mathcal{Z}([\widetilde{\bfw}_2,\widetilde{\vartheta}_2])
-\mathcal{Z}([\widetilde{\bfw}_1,\widetilde{\vartheta}_1])\|_{X_T}
\leq c K(T)\|[\widetilde{\bfw}_2,\widetilde{\vartheta}_2]
-[\widetilde{\bfw}_1,\widetilde{\vartheta}_1]\|_{X_T},
\end{equation}
where the positive function $K(T)\rightarrow 0_+$ for $T\rightarrow
0_+$ and $c$ does not depend on $T$. Hence taking $T^*\in (0,T]$
sufficiently small, the contraction of $\mathcal{Z}:B_R(T^*)
\rightarrow B_R(T^*)$ can be easily established. The proof of Lemma
\ref{ball_contraction} is complete.
\end{proof}

%%%%%%%%%%%%%%%%%%%%%%%%%%%%%%%%%%%%%%%%%%%%%%%%%%%%%%%%%%%%%%%%%%%%%%%%%%%%

\subsection{Global uniqueness of the strong solutions}
\label{sec_uniqueness}
Here we prove the global uniqueness of the strong solution stated in
the main result. Suppose that all assumptions of Theorem
\ref{theorem_main_result} are satisfied and there are two strong
solutions $[\bfu_1,\theta_1],[\bfu_2,\theta_2]$ of
\eqref{eq3}--\eqref{init_temp} on $(0,T)$. Denote
$\bfz=\bfu_1-\bfu_2\in X^u_T$ and $\sigma=\theta_1 - \theta_2\in
X^{\theta}_T$. Then $\bfz$ and $\sigma$ satisfy the equations
\begin{eqnarray}
(\bfz_t,\bfv) + a_u(\bfz,\bfv)
+b(\bfz,\bfu_2,\bfv)+b(\bfu_1,\bfz,\bfv)-(\bff \sigma,\bfv)&=&0,
\\
(\sigma_t,\varphi)_{\Omega}+ a_{\theta}(\sigma,\varphi)+
d(\bfz,\theta_1,\varphi) +
d(\bfu_2,\sigma,\varphi)-e(\bfz,\bfu_1,\varphi)
\nonumber\\
 - e(\bfu_2,\bfz,\varphi)
-(\sigma\bff\cdot\bfu_1,\varphi)_{\Omega}-(\theta_2\bff\cdot\bfz,\varphi)_{\Omega}
 &=& 0
\end{eqnarray}
for every $[\bfv,\varphi]\in \mathbf{V}_{u}^{1,2} \times
V_{\theta}^{1,2}$ and for almost every $t\in(0,T)$ and
$\bfz(0)={\bf0}$ and $\sigma(0)=0$. Hence substituting $\bfv=\bfz$
and $\varphi=\sigma$ we obtain estimates
\begin{eqnarray}\label{est101}
\frac{1}{2}\frac{{\rm d}}{{\rm d}
t}\|\bfz(t)\|^2_{\mathbf{V}_{u}^{0,2}} +
\|\bfz(t)\|^2_{\mathbf{V}_{u}^{1,2}} &\leq&
|b(\bfu_1(t),\bfz(t),\bfz(t))| \nonumber
\\
&&+|b(\bfz(t),\bfu_2(t),\bfz(t))| +|(\bff(t) \sigma(t),\bfz(t))|
\nonumber
\\
\end{eqnarray}
and
\begin{eqnarray}\label{est102}
\frac{1}{2}\frac{{\rm d}}{{\rm d}
t}\|\sigma(t)\|^2_{V_{\theta}^{0,2}} +
\|\sigma(t)\|^2_{V_{\theta}^{1,2}} &\leq&
|d(\bfz(t),\theta_1(t),\sigma(t))| +
|d(\bfu_2(t),\sigma(t),\sigma(t))|
\nonumber\\
&& +|e(\bfz(t),\bfu_1(t),\sigma(t))|
+|e(\bfu_2(t),\bfz(t),\sigma(t))|
\nonumber\\
&& +|(\sigma(t)\bff(t)\cdot\bfu_1(t),\sigma(t))_{\Omega}|
\nonumber\\
&&+|(\theta_2(t)\bff(t)\cdot\bfz(t),\sigma(t))_{\Omega}|
\end{eqnarray}
for a.e. $t \in (0,T)$.
To estimate term by term on the right-hand sides of \eqref{est101}
and \eqref{est102} the Gagliardo--Nirenberg interpolation
inequalities (cf. \cite[Theorem 5.8]{AdamsFournier1992})
\begin{eqnarray*}
\|\bfz(t)\|_{\mathbf{L}^4} &\leq& c\,
\|\bfz(t)\|^{1/2}_{\mathbf{W}^{1,2}}
\|\bfz(t)\|^{1/2}_{\mathbf{L}^2},
\\
\|\sigma(t)\|_{L^4(\Omega)} &\leq& c\,
\|\sigma(t)\|^{1/2}_{W^{1,2}(\Omega)}
\|\sigma(t)\|^{1/2}_{L^2(\Omega)}
\end{eqnarray*}
and the well-known Young's inequality with parameter $\delta$
$$
ab \leq \delta a^p + C(\delta)b^q \qquad (a,b>0, \;
\delta>0,\;1<p,q<\infty,\,1/p + 1/q = 1)
$$
for $C(\delta)=(\delta p)^{-q/p}q^{-1}$, will be frequently used.

Estimating the right-hand side of the inequality \eqref{est101} we
arrive at
\begin{eqnarray}\label{est103}
|b(\bfu_1(t),\bfz(t),\bfz(t))|&\leq&\|\bfu_1(t)\|_{\mathbf{L}^4}
\|\bfz(t)\|_{\mathbf{W}^{1,2}} \|\bfz(t)\|_{\mathbf{L}^{4}}
\nonumber\\
&\leq& c \|\bfu_1(t)\|_{\mathbf{L}^4}
\|\bfz(t)\|^{3/2}_{\mathbf{W}^{1,2}}
\|\bfz(t)\|^{1/2}_{\mathbf{L}^{2}}
\nonumber\\
&\leq& \delta\|\bfz(t)\|^{2}_{\mathbf{W}^{1,2}}
+C(\delta)\|\bfu_1(t)\|^4_{\mathbf{L}^4}\|\bfz(t)\|^{2}_{\mathbf{L}^{2}},
\end{eqnarray}
\begin{eqnarray}\label{est104}
|b(\bfz(t),\bfu_2(t),\bfz(t))| &\leq&
\|\bfu_2(t)\|_{\mathbf{W}^{1,2}} \|\bfz(t)\|^2_{\mathbf{L}^{4}}
\nonumber\\
 &\leq&
c\|\bfu_2(t)\|_{\mathbf{W}^{1,2}} \|\bfz(t)\|_{\mathbf{W}^{1,2}}
\|\bfz(t)\|_{\mathbf{L}^{2}}
\nonumber\\
&\leq& \delta\|\bfz(t)\|^{2}_{\mathbf{W}^{1,2}}
+C(\delta)\|\bfu_2(t)\|^2_{\mathbf{W}^{1,2}}\|\bfz(t)\|^{2}_{\mathbf{L}^{2}}
\end{eqnarray}
and
\begin{multline}\label{est105}
|(\bff(t) \sigma(t),\bfz(t))|  \leq \|\bff(t)\|_{\mathbf{L}^2}
\|\sigma(t)\|_{L^4(\Omega)}\|\bfz(t)\|_{\mathbf{L}^4}
\\
\leq c\|\bff(t)\|_{\mathbf{L}^2} \|\sigma(t)\|^{1/2}_{L^2(\Omega)}
\|\sigma(t)\|^{1/2}_{W^{1,2}(\Omega)}\|\bfz(t)\|^{1/2}_{\mathbf{L}^2}
\|\bfz(t)\|^{1/2}_{\mathbf{W}^{1,2}}
\\
\leq \delta\left( \|\sigma(t)\|^{2}_{W^{1,2}(\Omega)}
+\|\bfz(t)\|^{2}_{\mathbf{W}^{1,2}} \right)
+C(\delta)\|\bff(t)\|^2_{\mathbf{L}^2}\left(
\|\sigma(t)\|^{2}_{L^2(\Omega)} +\|\bfz(t)\|^{2}_{\mathbf{L}^2}
\right).  \\
\end{multline}
Now the estimates \eqref{est101} and \eqref{est103}--\eqref{est105}
imply
\begin{multline}\label{est120}
\frac{1}{2}\frac{{\rm d}}{{\rm d}
t}\|\bfz(t)\|^2_{\mathbf{V}_{u}^{0,2}} +
\|\bfz(t)\|^2_{\mathbf{V}_{u}^{1,2}} \leq  \delta\left(
\|\sigma(t)\|^{2}_{W^{1,2}(\Omega)}
+\|\bfz(t)\|^{2}_{\mathbf{W}^{1,2}} \right)
\\
+C(\delta)\left(\|\bff(t)\|^2_{\mathbf{L}^2}+\|\bfu_1(t)\|^4_{\mathbf{L}^4}
+\|\bfu_2(t)\|^2_{\mathbf{W}^{1,2}}\right)
\|\bfz(t)\|^{2}_{\mathbf{L}^2}
\\
+C(\delta)\|\bff(t)\|^2_{\mathbf{L}^2}
\|\sigma(t)\|^{2}_{L^2(\Omega)}.
\end{multline}
Similarly, let us estimate all terms on the right-hand side of
\eqref{est102}. Successively
\begin{multline}\label{est201}
|d(\bfz(t),\theta_1(t),\sigma(t))| \leq
\|\bfz(t)\|_{\mathbf{L}^4}\|\nabla\theta_1(t)\|_{\mathbf{L}^2}
\|\sigma(t)\|_{L^4(\Omega)}
\\
\leq c\|\theta_1(t)\|_{W^{1,2}(\Omega)}
\|\sigma(t)\|^{1/2}_{L^2(\Omega)}
\|\sigma(t)\|^{1/2}_{W^{1,2}(\Omega)}\|\bfz(t)\|^{1/2}_{\mathbf{L}^2}
\|\bfz(t)\|^{1/2}_{\mathbf{W}^{1,2}}
\\
\leq \delta\left( \|\sigma(t)\|^{2}_{W^{1,2}(\Omega)}
+\|\bfz(t)\|^{2}_{\mathbf{W}^{1,2}} \right)
\\
+C(\delta)\|\theta_1(t)\|^2_{W^{1,2}(\Omega)} \left(
\|\sigma(t)\|^{2}_{L^2(\Omega)} +\|\bfz(t)\|^{2}_{\mathbf{L}^2}
\right)
\end{multline}
and
\begin{eqnarray}\label{est202}
|d(\bfu_2(t),\sigma(t),\sigma(t))| &\leq&
\|\bfu_2(t)\|_{\mathbf{L}^4}\|\nabla\sigma(t)\|_{\mathbf{L}^2}
\|\sigma(t)\|_{L^4(\Omega)} \nonumber
\\
&\leq& c
\|\bfu_2(t)\|_{\mathbf{L}^4}\|\sigma(t)\|^{3/2}_{W^{1,2}(\Omega)}
\|\sigma(t)\|^{1/2}_{L^2(\Omega)} \nonumber
\\
&\leq& \delta \|\sigma(t)\|^{2}_{W^{1,2}(\Omega)}
+C(\delta)\|\bfu_2(t)\|^4_{\mathbf{L}^4}\|\sigma(t)\|^{2}_{L^2(\Omega)}.
\nonumber
\\
\end{eqnarray}
For the dissipative terms in \eqref{est102} we have
\begin{multline}\label{est203}
|e(\bfz(t),\bfu_1(t),\sigma(t))| \leq c
\|\bfz(t)\|_{\mathbf{W}^{1,2}}\|\bfu_1(t)\|_{\mathbf{W}^{1,4}}
\|\sigma(t)\|_{L^4(\Omega)}
\\
 \leq c \|\bfz(t)\|_{\mathbf{W}^{1,2}}\|\bfu_1(t)\|_{\mathbf{W}^{1,4}}
 \|\sigma(t)\|^{1/2}_{W^{1,2}(\Omega)}
\|\sigma(t)\|^{1/2}_{L^2(\Omega)}
\\
 \leq \delta\|\bfz(t)\|^2_{\mathbf{W}^{1,2}}
 + C(\delta)\|\bfu_1(t)\|^2_{\mathbf{W}^{1,4}}
 \|\sigma(t)\|_{W^{1,2}(\Omega)}
\|\sigma(t)\|_{L^2(\Omega)}
\\
 \leq \delta\left( \|\bfz(t)\|^2_{\mathbf{W}^{1,2}}
 + \|\sigma(t)\|^2_{W^{1,2}(\Omega)} \right)
 + C(\delta)^3\|\bfu_1(t)\|^4_{\mathbf{W}^{1,4}}
\|\sigma(t)\|^2_{L^2(\Omega)}
\\
\end{multline}
and in the same way we deduce the inequality
\begin{multline}\label{est204}
|e(\bfu_2(t),\bfz(t),\sigma(t))|\leq c
\|\bfu_2(t)\|_{\mathbf{W}^{1,4}}\|\bfz(t)\|_{\mathbf{W}^{1,2}}
\|\sigma(t)\|_{L^4(\Omega)}
\\
 \leq \delta\left( \|\bfz(t)\|^2_{\mathbf{W}^{1,2}}
 + \|\sigma(t)\|^2_{W^{1,2}(\Omega)} \right)
 + C(\delta)^3\|\bfu_2(t)\|^4_{\mathbf{W}^{1,4}}
\|\sigma(t)\|^2_{L^2(\Omega)}.
\\
\end{multline}
The last two terms can be estimated as follows:
\begin{multline}\label{est205}
|(\sigma(t)\bff(t)\cdot\bfu_1(t),\sigma(t))_{\Omega}| \leq
\|\bff(t)\|_{\mathbf{L}^{2}}\|\bfu_1(t)
\|_{\mathbf{L}^{\infty}}\|\sigma(t)\|^2_{L^4(\Omega)}
\\
\leq c
\|\bff(t)\|_{\mathbf{L}^{2}}\|\bfu_1(t)\|_{\mathbf{L}^{\infty}}
\|\sigma(t)\|_{W^{1,2}(\Omega)} \|\sigma(t)\|_{L^2(\Omega)}
\\
\leq  \delta \|\sigma(t)\|^2_{W^{1,2}(\Omega)}
+C(\delta)\|\bff(t)\|^2_{\mathbf{L}^{2}}\|\bfu_1(t)\|^2_{\mathbf{L}^{\infty}}
\|\sigma(t)\|^2_{L^2(\Omega)}
\end{multline}
and finally
\begin{multline}\label{est206}
|(\theta_2(t)\bff(t)\cdot\bfz(t),\sigma(t))_{\Omega}|\leq
\|\theta_2(t)\|_{L^{\infty}(\Omega)}\|\bff(t)\|_{\mathbf{L}^{2}}
\|\bfz(t)\|_{\mathbf{L}^{4}}\|\sigma(t)\|_{L^4(\Omega)}
\\
\leq c
\|\theta_2(t)\|_{L^{\infty}(\Omega)}\|\bff(t)\|_{\mathbf{L}^{2}}
\|\bfz(t)\|^{1/2}_{\mathbf{W}^{1,2}}
\|\bfz(t)\|^{1/2}_{\mathbf{L}^{2}}
\|\sigma(t)\|^{1/2}_{W^{1,2}(\Omega)}
\|\sigma(t)\|^{1/2}_{L^2(\Omega)}
\\
\leq  \delta
\|\bfz(t)\|_{\mathbf{W}^{1,2}}\|\sigma(t)\|_{W^{1,2}(\Omega)}
+C(\delta)\|\theta_2(t)\|^2_{L^{\infty}(\Omega)}\|\bff(t)\|^2_{\mathbf{L}^{2}}
\|\bfz(t)\|_{\mathbf{L}^{2}}\|\sigma(t)\|_{L^2(\Omega)}
\\
\leq \frac{\delta}{2}\left( \|\sigma(t)\|^{2}_{W^{1,2}(\Omega)}
+\|\bfz(t)\|^{2}_{\mathbf{W}^{1,2}} \right)
\\
+C(\delta)\|\theta_2(t)\|^2_{L^{\infty}(\Omega)}\|\bff(t)\|^2_{\mathbf{L}^{2}}
\left( \|\sigma(t)\|^{2}_{L^2(\Omega)}
+\|\bfz(t)\|^{2}_{\mathbf{L}^2} \right).
\end{multline}
Now the estimates \eqref{est102} and \eqref{est201}--\eqref{est206}
imply
\begin{multline}\label{est220}
\frac{1}{2}\frac{{\rm d}}{{\rm d}
t}\|\sigma(t)\|^2_{V_{\theta}^{0,2}} +
\|\sigma(t)\|^2_{V_{\theta}^{1,2}} \leq
\delta\left( \|\bfz(t)\|^{2}_{\mathbf{W}^{1,2}} +
\|\sigma(t)\|^{2}_{W^{1,2}(\Omega)}\right)
\\
+C(\delta)\left(\|\theta_1(t)\|^2_{W^{1,2}(\Omega)}
+\|\theta_2(t)\|^2_{L^{\infty}(\Omega)}\|\bff(t)\|^2_{\mathbf{L}^{2}}\right)
\|\bfz(t)\|^{2}_{\mathbf{L}^2}
\\
+C(\delta)\left(\|\theta_1(t)\|^2_{W^{1,2}(\Omega)}
+\|\bfu_2(t)\|^4_{\mathbf{L}^4}+C(\delta)^2\|\bfu_1(t)\|^4_{\mathbf{W}^{1,4}}
+C(\delta)^2\|\bfu_2(t)\|^4_{\mathbf{W}^{1,4}} \right.
\\
\left.
+\|\bff(t)\|^2_{\mathbf{L}^{2}}\|\bfu_1(t)\|^2_{\mathbf{L}^{\infty}}
+\|\theta_2(t)\|^2_{L^{\infty}(\Omega)}\|\bff(t)\|^2_{\mathbf{L}^{2}}
\right) \|\sigma(t)\|^{2}_{L^2(\Omega)}.
\end{multline}
Choosing $\delta$ sufficiently small and summing \eqref{est120} and
\eqref{est220} we conclude that
\begin{equation}\label{}
\frac{{\rm d}}{{\rm d} t} \left(\|\bfz(t)\|^2_{\mathbf{V}_{u}^{0,2}}
+ \|\sigma(t)\|^2_{V_{\theta}^{0,2}} \right) \leq \chi(t)
\left(\|\bfz(t)\|^2_{\mathbf{V}_{u}^{0,2}} +
\|\sigma(t)\|^2_{V_{\theta}^{0,2}} \right),
\end{equation}
where $\chi(t)\in L^1((0,T))$. Now the uniqueness follows from the
fact that $\bfz(0)={\bf0}$ and $\sigma(0)=0$ using the Gronwall
lemma.

%%%%%%%%%%%%%%%%%%%%%%%%%%%%%%%%%%%%%%%%%%%%%%%%%%%%%%%%%%%%%%%%%%%%%%



\end{document}